\documentclass[11pt,leqno]{amsart}
\usepackage[utf8]{inputenc}
\usepackage{graphicx,amsmath,amssymb,bm,amsthm,comment,color,appendix,a4wide,bbm}
\usepackage{hyperref}
\hypersetup{
  colorlinks = true,
  linkcolor  = blue,
  citecolor    = red
}
\usepackage{ulem}

\graphicspath{{./Figures/}}
\newcommand{\N}{\mathbb{N}}

\newcommand{\R}{\mathbb{R}}
\newcommand{\C}{\mathbb{C}}

\newcommand{\mc}[1]{\mathcal{#1}}
\newcommand{\p}{\partial}
\newcommand{\e}{\varepsilon}
\newcommand{\dl}{\delta}
\newcommand{\pphi}{\varphi}

\newcommand{\ord}{{\mathcal O}}

\newcommand{\cW}{{\mathcal W}}
\newcommand{\ope}[1]{\operatorname{#1}}
\newcommand{\til}[1]{{\widetilde{#1}}}
\newcommand{\be}{\begin{equation}}
\newcommand{\ee}{\end{equation}}
\newcommand{\ben}{\begin{equation*}}
\newcommand{\een}{\end{equation*}}

\newcommand{\dip}{\displaystyle}
\newcommand{\qtext}[1]{\quad\text{#1 }\ }

\newcommand{\dir}{{\bullet}}
\newcommand{\inc}{{\flat}}
\newcommand{\out}{{\sharp}}
\theoremstyle{theorem}
\newtheorem{theorem}{Theorem}
\newtheorem{lemma}{Lemma}[section]
\newtheorem{proposition}[lemma]{Proposition}

\newtheorem{remark}[lemma]{Remark}

\numberwithin{equation}{section}
\newcounter{Condition}

\newtheorem{condition}[Condition]{Condition}
\makeatletter
\@namedef{subjclassname@2020}{%
  2020 Mathematics Subject Classification}
\makeatother


\title[Local scattering matrix for a degenerate avoided-crossing]{
Local scattering matrix for a degenerate avoided-crossing in the non-coupled regime
}
\author{Kenta Higuchi}
\address[Kenta Higuchi]{Graduate School of Science and Engineering, Ehime University/ Bunkyocho 3, Matsuyama, Ehime, 790-8577, Japan.
 e-mail: higuchi.kenta.vf@ehime-u.ac.jp}
\date{}

\begin{document}

\maketitle
\begin{abstract}
A Landau-Zener type formula for a degenerate avoided-crossing is studied in the non-coupled regime. More precisely, a $2\times2$ system of first order $h$-differential operator with $\mathcal{O}(\varepsilon)$ off-diagonal part is considered in 1D. 
Asymptotic behavior as $\e h^{m/(m+1)}\to0^+$ of the local scattering matrix near an avoided-crossing is given, where $m$ stands for the contact order of two curves of the characteristic set.
A generalization including the cases with vanishing off-diagonals and non-Hermitian symbols is also given.

\end{abstract}

\section{Introduction}
\subsection{Background}
We study the following system with small parameters $\e,h>0$ in 1D:
\be\label{eq:System}
P(\e,h)w:=(hD_x\otimes \ope{Id}_2+H(\e))w=0,\quad
H(\e)=\begin{pmatrix}V_1(x)&\e\\\e&V_2(x)\end{pmatrix},\quad D_x=-i\frac{d}{dx}.
\ee
For a linear function $V_1(x)=-V_2(x)=vx$ $(v>0)$, this problem was introduced in the context of the adiabatic theory \cite{La,Ze}. In their model, the variable $x$ is interpreted as time, the eigenvalues $\pm\sqrt{V(x)^2+\e^2}$ of $H(\e)$ as possible energies, and the small parameters $\e$ and $h$ as an interaction between two energies and time scale, respectively. They computed the transition probability $\exp(-\frac{\pi}{v} \e^2h^{-1})$ between two possible energies while the time $x$ varies from $-\infty$ to $+\infty$. This is the well-known Landau-Zener formula. 

The local scattering matrix introduced in \cite{CdVLP} is a microlocal object which makes this global problem on the transition probability into a microlocal problem. They consider a general situation for matrix pseudodifferential operators with a transversal crossing. For our model operator $P(h,\e)$,
the diagonal matrix-valued function
\be\label{eq:p0}
p_0(x,\xi)=\ope{diag}(\xi+V_1(x),\xi+V_2(x))=:\ope{diag}(p_1(x,\xi),p_2(x,\xi))
\ee
is the symbol 
when $\e=0$. There are two curves $\Gamma_1$ and $\Gamma_2$ defined by $\Gamma_j=p_j^{-1}(0)=\{(x,\xi)\in T^*\R;\,\xi+ V_j(x)=0\}$ $(j=1,2)$. They cross at $(x_0,-V_0)$ provided that $V_1(x_0)=V_2(x_0)=V_0$ for an $x_0\in\R$. The crossing is transversal if  $V_1'(x_0)\neq V_2'(x_0)$. On each curve $\Gamma_j$, we consider the Hamiltonian flow induced by the classical Hamiltonian $p_j$, that is, the flow of the Hamiltonian vector field in $T(T^*\R)$:
\be
H_{p_j}=\frac{\p p_j}{\p \xi}\p_x-\frac{\p p_j}{\p x}\p_\xi=\p_x-V_j'(x)\p_\xi.
\ee
The local scattering matrix describes the behavior of a (microlocal) solution on the outgoing curves $\gamma_{1,r}$ and $\gamma_{2,r}$ from the crossing point $(x_0,-V_0)$ in terms of that on the incoming curves $\gamma_{1,\ell}$ and $\gamma_{2,\ell}$, where we put $\gamma_{j,\ell}=\Gamma_j\cap\{x<x_0\}$ and $\gamma_{j,r}=\Gamma_j\cap\{x>x_0\}$.  
The transition probability, in particular the Landau-Zener formula, appears as the square of the modulus of the diagonal part of this unitary $2\times2$-matrix \cite[Formula (9)]{CdV3}, \cite[Formula (8)]{CdVLP}. Similar objects to this matrix have been applied to study the two-level adiabatic transition probability with several avoided-crossings \cite{Jo2,JMP,Wa,WaZe}, behavior of the (global) scattering matrix \cite{ABA} and asymptotic repartition of eigenvalues and resonances \cite{AsFu,FMW3,Hi1} of coupled Schr\"odinger operators, etc. 

Our main objective in this manuscript is to compute the local scattering matrix in the case that the two curves $\Gamma_1$ and $\Gamma_2$ contacts tangentially, that is, $V_1'(x_0)=V_2'(x_0)$. Such a problem has been considered in the \textit{coupled regime} (named by \cite{CdVLP}) in \cite{Wa}, that is, the parameter $h$ is sufficiently small compared to $\e$. We consider the opposite regime named \textit{non-coupled regime}. Non-coupled regime is studied for transversal crossings in \cite{CdV3,CdVLP,WaZe}, where the regimes are also called \textit{adiabatic regime} and \textit{non-adiabatic regime} in \cite{WaZe}. Let us make precise the notion of regimes. In the transversal case, the transition probability is $\exp(-c\e^2h^{-1})$ with a positive $c>0$ \cite{CdVLP,Jo2} like the Landau-Zener formula. This is exponentially small when $\e^2\gg h$, the coupled regime for this case, whereas it admits $1+\ord(\e^2h^{-1})$ when $\e^2\ll h$, the non-coupled regime. The critical rate of the exponent may vary for the cases with a tangential crossing. In \cite{Wa}, the regime that $\e^{m+1}\gg h^m$ is studied, where $m$ stands for the contact order, that is,
\be\label{eq:contact-o}
V_1^{(k)}(x_0)-V_2^{(k)}(x_0)=0\quad(0\le k\le m-1),\qquad
V_1^{(m)}(x_0)-V_2^{(m)}(x_0)\neq0.
\ee
They obtain the exponentially small transition probability $\exp(-c \e^{\frac{m+1}m}h^{-1})$ with $c>0$. 
In this manuscript, we show asymptotic behavior of the local scattering matrix in the regime that $\e^{m+1}\ll h^m$. We conclude from these results that $\e^{m+1}\gg h^m$ and $\e^{m+1}\ll h^m$ are coupled and non-coupled regimes for this situation. 

A problem with the tangential crossing has been also studied by the author with his collaborators for a model of coupled Schr\"odinger operators \cite{AFH1,AFH2}. 
They 
showed asymptotics of the local scattering matrix and resonances. 
Our first order model is simpler than theirs, however, asymptotic formulae for the local scattering matrix for our model and for that for their model can be written in a single formula in terms of the symbol (see Formulae \eqref{eq:Common-T1} and \eqref{eq:Common-T2}).
This suggests a possibility for a generalization as well as the transversal settings  (see \cite{CdV,CdV3} for a matrix normal form, and \cite{AsFu} for an application 
of the scalar normal form via a reduction of \cite{HeSj}).
However, our proof is based on the peculiarity of the model (sum of a function of $x$ and a polynomial of $\xi$), and can not be generalized immediately.

In this section, we state our main result (Subsection~\ref{Sec:Result}), discuss on the consistency with other studies (Subsection~\ref{Sec:Consistent}), and explain the plan of the manuscript (Subsection~\ref{Sec:Plan}).

\subsection{Assumptions and main result}\label{Sec:Result}
We here state our results. We will make precise in Subsection~\ref{sec:Terminologies} the terminologies of semiclassical and microlocal analysis used for stating the results. 
We consider the system \eqref{eq:System} locally in an interval $I=]x_\ell,x_r[\subset\R$ near $x=0$, where $\e,h$ are positive parameters,  and $V_1$, $V_2$ are smooth $(C^\infty)$ functions satisfying the following conditions:
\begin{condition}\label{C1}
The functions $V_1$, $V_2$ are real-valued.
$V_1-V_2$ vanishes only at $x=0$. The vanishing order there is finite.
\end{condition}

Put $V_0:=V_1(0)=V_2(0)$ and $\rho_0:=(0,-V_0)$. 
Let us denote by $m\ge1$ the vanishing order at $x=0$ of $V_1-V_2$ (see \eqref{eq:contact-o}). 
For $k\in\N=\{0,1,2,\ldots\}$, put
\be\label{eq:mu-k}
\mu_k=\mu_k(\e,h)
=\e h^{-\frac k{k+1}}.
\ee
Note that for $0<h<1$, $k_1>k_2$ implies $\mu_{k_1}>\mu_{k_2}$. 
The function $p_0$ given by \eqref{eq:p0} is the principal symbol of $P$ if $\mu_1\to0$ as $h\to0^+$, in the sense that the semiclassical wavefront set of a microlocal solution is contained in the characteristic set $\Gamma=\{(x,\xi)\in T^*\R;\,\det p_0(x,\xi)=0\}$. 
Note that in $I_x\times\R_\xi\subset T^*\R_{(x,\xi)}$, the kernel of the matrix $p_0(x,\xi)$ is zero, one, and two-dimensional at each point of $\Gamma^c$, $\Gamma\setminus\{\rho_0\}$, and $\{\rho_0\}$, respectively. 
This coincides with the dimension of the space of microlocal solutions near each point. It is well-known for each point of $\Gamma^c$ and $\Gamma\setminus\{\rho_0\}$ (see Subsection~\ref{sec:Terminologies}), but not for $\{\rho_0\}$ (see Proposition~\ref{prop:2dim}).  We divide $\Gamma\setminus\{\rho_0\}$ into four parts (see Figure~\ref{Fig:CS})
\ben
\Gamma\setminus\{\rho_0\}=\gamma_{1,\ell}\cup\gamma_{1,r}\cup\gamma_{2,\ell}\cup\gamma_{2,r},
\een
with
\ben
\gamma_{j,\ell}=\Gamma_j\cap\{x<0\},\quad\gamma_{j,r}=\Gamma_j\cap\{x>0\},
\quad
\Gamma_j=\{(x,\xi)\in T^*\R;\,p_j(x,\xi)=0\}\quad(j=1,2).
\een
\begin{figure}
\centering
\includegraphics[bb=0 0 371 198, width=8cm]{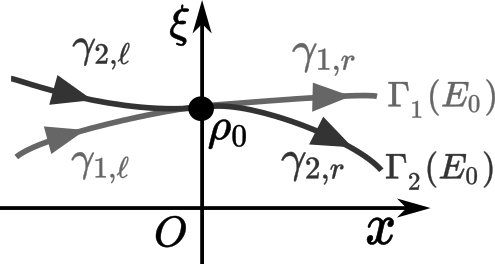}
\caption{Characteristic set of $p_0$ near the crossing point}
\label{Fig:CS}
\end{figure}
On each $\gamma_{j,\dir}$, the system \eqref{eq:System} is reduced to a single equation, and the space of microlocal solutions near $\gamma_{j,\dir}$ is one-dimensional. 
Moreover, there exists a function $f_{j,\dir}$ which spans the space and admits the asymptotic expansion
\be\label{eq:mlBasis-f}
f_{j,\ell}=e^{-i\int_0^xV_j(y)dy/h}
\begin{pmatrix}a_{j,\ell}\\
b_{j,\ell}
\end{pmatrix}\quad\text{for }\ x<0,\quad
f_{j,r}=
e^{-i\int_0^xV_j(y)dy/h}
\begin{pmatrix}a_{j,r}\\
b_{j,r}
\end{pmatrix}\quad\text{for }\ x>0,
\ee
with smooth functions $a_{j,\dir}$ and $b_{j,\dir}$ whose asymptotic behavior is given by
\be\label{eq:mlBasis-fab}
\begin{aligned}
&a_{1,\dir}=1+\ord(\mu_1^2),
&&b_{1,\dir}=\ord(\e),\\
&a_{2,\dir}=\ord(\e),
&&b_{2,\dir}=1+\ord(\mu_1^2),
\end{aligned}\qquad(\dir=\ell,r).
\ee
\begin{remark}
Such a function $f_{j,\dir}$ satisfying \eqref{eq:mlBasis-f} and \eqref{eq:mlBasis-fab} is unique up to an error of $\ord(\mu_1^2)$. More precisely, if $f_{j,\dir}$ and $\til{f}_{j,\dir}$ are two functions which satisfy the conditions, then there exist $C=C(\e,h)$ bounded as $h\to0^+$ such that 
\be\label{eq:f-modulo}
\til{f}_{j,\dir}=(1+\mu_1^2C)f_{j,\dir}.
\ee 
Note that when $\e=h$, $f_{j,\dir}$'s  are the standard microlocal solutions:
\ben
f_{1,\dir}=\exp\left(\frac1{ih}\int_0^x V_1(y)dy\right)
\begin{pmatrix}1+\ord(h)\\\ord(h)\end{pmatrix},\qquad
f_{2,\dir}=
\exp\left(\frac1{ih}\int_0^x V_2(y)dy\right)
\begin{pmatrix}\ord(h)\\1+\ord(h)\end{pmatrix}.
\een
\end{remark}

Let $\psi$ be a microlocal solution to \eqref{eq:System} near $\rho_0$. 
Since the space of the microlocal solutions near each $\gamma_{j,\dir}$ $(j,\dir)\in\{1,2\}\times\{\ell,r\}$ is one-dimensional, there exist constants $\alpha_{1,\ell}$, $\alpha_{2,\ell}$, $\alpha_{1,r}$ and $\alpha_{2,r}$ such that
\be\label{eq:mlCoeff}
\psi\equiv \alpha_{j,\dir}f_{j,\dir}\qtext{microlocally near}\gamma_{j,\dir}.
\ee
Then Proposition~\ref{prop:2dim}, which states under Condition~\ref{C1} that the space of microlocal solutions near $\rho_0$ is two-dimensional when $\mu_m$ is sufficiently small, implies that there exists a $2\times2$-matrix such that 
\be\label{eq:Def-T}
\begin{pmatrix}\alpha_{1,r}\\\alpha_{2,r}\end{pmatrix}
=T\begin{pmatrix}\alpha_{1,\ell}\\\alpha_{2,\ell}\end{pmatrix}.
\ee
Note that the functions $f_{j,\dir}$ are characterized up to $\ord(\mu_1^2)$ by \eqref{eq:mlBasis-f} and \eqref{eq:mlBasis-fab}, hence the constants $\alpha_{j,\dir}$ and the matrix $T$ are also. More precisely, let $t_{jk}$ be the $(j,k)$-entry of $T$ for a choice of $(f_{1,\ell},f_{2,\ell},f_{1,r},f_{2,r})$. Then for another choice, the $(j,k)$-entry $s_{jk}$ of $T$ satisfies 
\be\label{eq:ess-error2}
s_{jk}=t_{jk}(1+\ord(\mu_1^2)).
\ee
This is a straightforward consequence of \eqref{eq:f-modulo}.

\begin{theorem}\label{thm:Asymptot-T}
Assume Condition~\ref{C1} with the vanishing order $m\ge2$. Then there exist $h_0>0$ and $\mu^0>0$ such that $T=T(h,\e)$ admits 
\be\label{eq:Asym-T1}
T=\begin{pmatrix}
1+\ord(\mu_m^2)
&t_{12}
\\t_{21}
&1+\ord(\mu_m^2)
\end{pmatrix}
\ee
with
\be\label{eq:t12-t21}
t_{12}=-i\mu_m(\omega_m(h)+\ord(\mu_m)),
\quad
t_{21}=-i\mu_m(\overline{\omega_m(h)}+\ord(\mu_m)),
\ee
uniformly for $0<\mu_m<\mu^0$ and $0<h<h_0$. 
Here, the asymptotic behavior of $\omega_m(h)$ is given by
\be\label{eq:omega-m}
\omega_m(h)=2\eta_m\left(\frac{(m+1)!}{\bigl|V_1^{(m)}(0)-V_2^{(m)}(0)\bigr|}\right)^{\frac1{m+1}}\bm{\Gamma}\left(\frac{m+2}{m+1}\right)+\ord(h^{\frac1{m+1}}),
\ee
where $\bm{\Gamma}$ stands for the Gamma function, and 
\be
\eta_m=\left\{
\begin{aligned}
&\cos\left(\frac\pi{2(m+1)}\right)&&\text{when $m$ is even,}\\
&\exp\left(\ope{sgn}(V_1^{(m)}(0)-V_2^{(m)}(0))\frac{i\pi}{2(m+1)}\right)&&\text{when $m$ is odd.}
\end{aligned}
\right.
\ee
Moreover, $T$ is unitary for a suitable choice of microlocal solutions $(f_{1,\ell},f_{2,\ell},f_{1,r},f_{2,r})$.
\end{theorem}
Note that the error $\ord(\mu_m^2)$ in the above formula is larger than or the same order as the essential error of order $\mu_1^2$ coming from the choice of bases of microlocal solutions.
\begin{remark}
The above statement is also true for transversal case $m=1$ (see \cite{CdV3,CdVLP,WaZe}). Our proof for the case $m=1$ however is valid only if $\mu_1\left(\log (1/h)\right)^{1/2}$ is sufficiently small whereas the smallness of $\mu_1$  is enough for their argument.
\end{remark}
\begin{remark}
The error of $\ord(h^{\frac1{m+1}})$ in \eqref{eq:omega-m} can be replaced by $\ord(h^{\frac2{m+1}})$ when $m$ is odd.
\end{remark}

\subsection{Consistency with other studies}\label{Sec:Consistent}
Let us compare with some studies on the pseudodifferential operator $Q=q^w(x,hD;\e)$ associated with a Hermitian matrix-valued symbol $q$ with the diagonal principal term:
\be
q(x,\xi;\e)=\begin{pmatrix}q_1(x,\xi)&0\\0&q_2(x,\xi)\end{pmatrix}+\e\begin{pmatrix}0&W(x,\xi)\\\overline{W(x,\xi)}&0\end{pmatrix},\quad
q_j(x,\xi)\in C^\infty(T^*\R;\R).
\ee
In the studies of matrix Schr\"odinger operators like \cite{ABA,AsFu,AFH1,AFH2,FMW3,Hi1}, $q_j=\xi^2+V_j(x)$ is the symbol of the scalar Schr\"odinger operator with the potential $V_j$, and $\e=h$ (note that in \cite{AsFu}, they considered not only scalar Schr\"odinger operators for $q_1^w$ and $q_2^w$). 
In the study of ``non-adiabatic" two-level transition in \cite{WaZe}, $q_j=\xi+V_j(x)$. 
In each of these studies, they computed the ``local scattering matrix" near a crossing point in the non-coupling regime.

Let $\rho\in \{q_1=q_2=0\}\subset T^*\R$ be a crossing point, and let $\gamma_j^\inc$ and $\gamma_j^\out$ be small subsets near $\rho$ of $\Gamma_j=\{q_j=0\}$ such that the Hamiltonian flow associated with $q_j$ is incoming to $\rho$ on $\gamma_j^\inc$ and outgoing from $\rho$ on $\gamma_j^\out$. On each $\Gamma_j$, $d_{x,\xi}q_j$ does not vanish, and the equation $q_j^w(x,hD_x)u=0$ is microlocally reduced to $hD_xv=0$ by a unitary transform $u=F_jv$ ($F_j$ is a Fourier integral operator associated with a symplectomorphism $\kappa_j$ with $q_j\circ\kappa_j(y,\eta)=\eta$ and $\kappa_j(0,0)=\rho$). We denote by $u_j^\dir:=(\chi_j^\dir)^w(x,hD_x)F_j\bm{1}$, where $\chi_j^\dir\in C_c^\infty(T^*\R;[0,1])$ is a cut-off function near $\gamma_j^\dir$, and $\bm{1}=1$ identically. 
Taking a tuple $(f_1^\inc,f_2^\inc,f_1^\out,f_2^\out)$ of bases of microlocal solutions on $(\gamma_1^\inc,\gamma_2^\inc,\gamma_1^\out,\gamma_2^\out)$ such that 
\ben
f_1^\dir(x;h)=
\begin{pmatrix}
u_1^\dir(x)+\ord(\mu_1^2)\\
\ord(\e)
\end{pmatrix},\quad
f_2^\dir(x;h)=
\begin{pmatrix}
\ord(\e)\\
u_2^\dir(x)+\ord(\mu_1^2)\\
\end{pmatrix}\qquad(\dir=\inc,\out),
\een
we can define the local scattering matrix modulo $\ord(\mu_1^2)$ as follows:

For each microlocal solution $\psi$ to $Qf=0$ near $\rho$, there exist $\alpha_1^\inc$, $\alpha_2^\inc$, $\alpha_1^\out$ and $\alpha_2^\out$ such that
\ben
\psi\equiv\alpha_j^\dir f_j^\dir\qtext{microlocally near}\gamma_j^\dir\qquad(j,\dir)\in\{1,2\}\times\{\inc,\out\}.
\een
There exists a $2\times2$-matrix $T=T(\e,h)$ such that
\be
\biggl(\begin{matrix}\alpha_1^\out\\\alpha_2^\out\end{matrix}\biggr)=T
\left(\begin{matrix}\alpha_1^\inc\\\alpha_2^\inc\end{matrix}\right).
\ee

The asymptotic behavior obtained in the above studies is roughly written in
\be\label{eq:Common-T1}
T\sim e^{-\frac{iM}2}\left(\ope{Id}_2-i
\mu_m T_{\ope{sub}}\right)
\ee
with $m$ the contact order of $\Gamma_1$ and $\Gamma_2$ at $\rho$, that is, the smallest integer such that
\be
H_{p_1}^mp_2(\rho)\neq0,\qquad(\text{equivalently }\ H_{p_2}^mp_1(\rho)\neq0),
\ee
$M$ coming from the Maslov index ($M=1$ if $\rho$ is a turning point and $M=0$ otherwise), and $T_{\ope{sub}}$ has the form
\be
T_{\ope{sub}}=\begin{pmatrix}
0&\overline{\omega}\\\omega&0
\end{pmatrix}
\ee
with
\be\label{eq:Common-T2}
\omega=2\eta_{m}\overline{W(\rho)}\bm{\Gamma}\left(\frac{m+2}{m+1}\right)\left(\frac{(m+1)!}{\left|H_{p_1}^mp_2(\rho)\right|}\right)^{\frac1{m+1}}.
\ee
Note that only the model in \cite{AsFu} allows the difference between $\p_\xi q_1$ and $\p_\xi q_2$, and their formula includes these values. 
\subsection{Terminologies of semiclassical and microlocal analysis}\label{sec:Terminologies}
We recall some terminologies of semiclassical and microlocal analysis (see e.g., \cite{DiSj,Ma,Zw} for more details). 
For a symbol $a\in C_c^\infty(T^*\R;\C^{N\times N})$ $(N\in\N\setminus\{0\})$, the corresponding $h$-pseudodifferential operator, denoted by $a^w(x,hD_x)$, is defined as a bounded operator in $L^2(\R;\C^2)$ via the $h$-Weyl quantization
\begin{align}\label{WeylQ}
 a^w(x,hD_x)u(x)
:=\frac{1}{2\pi h}\int_{T^*\R}e^{i(x-y)\xi/h}a\left(\frac{x+y}{2},\xi\right)u(y)dyd\xi.
\end{align}
Let $(x_0,\xi_0)\in T^*\R$ and $f=f(x;h)\in L^2(\R; \C^N)$ with $\|f\|_{L^2}=\ord(h^{-k})$ for some $k\in\N$ and for any small $h>0$. 
We say that $f$ is microlocally $0$ near $(x_0,\xi_0)$ and denote that by
\ben
f(x;h)\equiv0\qtext{near}(x_0,\xi_0),
\een
if there exists a symbol $\chi\in C_c^\infty(T^*\R;\C^{N\times N})$ with $\det\, \chi(x_0,\xi_0)\neq0$ such that 
\ben
\| \chi^w (x,hD_x)f\|_{L^2(\R)}=\ord(h^\infty).
\een 
We also say that $f$ is microlocally $0$ near a set $\Omega\subset T^*\R$ if it is microlocally $0$ near each point of $\Omega$. 
We define $\ope{WF}_h(f)$, the semiclassical wave front set of $f$, as the set of all points of $T^*\R$ where $f$ is not microlocally $0$. For an $h$-pseudodifferential operator $a^w$, we say a function $f(\,\cdot\,;h)\in L^2$ is a microlocal solution to the equation $a^wf=0$ near $(x_0,\xi_0)\in T^*\R$ if there exists $k\in\N$ such that
\ben
\left\|f(\,\cdot\,;h)\right\|_{L^2}=\ord(h^{-k})\qtext{as}h\to0^+,\qtext{and}
a^wf(x;h)\equiv0\qtext{near}(x_0,\xi_0).
\een

The following results for scalar $h$-pseudodifferential operators are well-known (see e.g. \cite[Theorem 12.5]{Zw}).  Let $a\in C_c^\infty(T^*\R;\C)$.  If $f$ is a microlocal solution on $\Omega\subset T^*\R$ to $a^wf=0$,  we have $\ope{WF}_h(f)\cap \Omega\subset \{a=0\}$. 
Moreover, if $\p a\neq0$ on $\{a=0\}\cap \Omega$, the semiclassical wavefront set $\ope{WF}_h(f)\cap \Omega$ is invariant under the Hamiltonian flow of $a$. 

\subsection{Plan of the manuscript}\label{Sec:Plan}
In the next section, we generalize Theorem~\ref{thm:Asymptot-T} to a slightly wider class of systems. 
Then, we see Theorem~\ref{thm:Asymptot-T} as a particular result of Theorem~\ref{thm:Asymptot-T2} with Proposition~\ref{thm:Unitary-SM1}. We devote Sections~\ref{Sec:Proof1} and \ref{Sec:Proof2} to the proof of these results. 
In Section~\ref{Sec:Proof1}, we compute the asymptotics for the transfer matrix between two bases of (exact) solutions to the system without any microlocal argument. A justification by  microlocal arguments is done in Section~\ref{Sec:Proof2}. The unitarity stated in Propositions~\ref{thm:Unitary-SM1} and \ref{thm:Unitary-SM2} is also shown in Section~\ref{Sec:Proof2}.

\section{Generalization}
Our method is applicable for a slightly more general class of first-order differential systems. 
One of developments is that we allow the determinant of the full-symbol to vanish. 
Such a situation is considered in \cite{CdV3} for a transversal crossing $(m=1)$, and in \cite{Jo} for a complex crossing, that is, when the determinant has a non-real complex zero. The recent result \cite{Lo} on resonance  asymptotics for a matrix Schr\"odinger operator also concerned this vanishment. The crossing there of possible energies is no longer avoided, and the border of the regimes also depends on the vanishing order of the subprincipal term. 
Another one is that the symbol does not necessarily be Hermitian matrix-valued. 
We show that the local scattering matrix becomes unitary when the system is equivalently given by an operator with a Hermitian matrix-valued symbol.

\subsection{Setting and asymptotics for the transfermatrix}
We consider the system
\be\label{eq:System2}
P(\e_1,\e_2,h)w:=\left(hD_x\otimes \ope{Id}_2+H(\e_1,\e_2)\right)w=0,\quad
H(\e_1,\e_2)=\begin{pmatrix}V_1&\e_1U_1\\\e_2U_2&V_2\end{pmatrix},
\ee
where $\e_1,\e_2,h$ are positive parameters,  
and functions $V_1$, $V_2$, $U_1$, and $U_2$ are smooth $(C^\infty)$ satisfying Condition~\ref{C1} and the following.
\begin{condition}\label{C2}
The vanishing order at $x=0$ of $U_1$ and of $U_2$ are at most finite.
\end{condition}

For $j=1,2$, we denote by $n_j\in\N=\{0,1,2,\ldots\}$ the vanishing order at $x=0$ of $U_j$:
\ben
U_j^{(k)}(0)=0,\quad(0\le k\le n_j-1)\qquad
U_j^{(n_j)}(0)\neq0.
\een
We continue to use the notations introduced in the previous section. 
The microlocal properties away from the crossing point $\rho_0=(0,-V_0)$ is similar. We denote by 
\ben
\til{\e}=\sqrt{\e_1\e_2}\quad\text{and }\ \til{n}=\frac12(n_1+n_2)
\een
the geometric mean of $\e_1$, $\e_2$ and the arithmetic mean of $n_1$, $n_2$, respectively. The condition
\be\label{eq:weak-regime2}
\til{\mu}_1:= \mu_1(\til{\e}
,h)=\sqrt{\e_1\e_2}h^{-1/2}\to0\qquad\text{as }\ h\to0^+
\ee
is enough for recognizing $\Gamma=p_0^{-1}(0)$ as the characteristic set. 
For each 
$(j,\dir)\in\{1,2\}\times\{\ell,r\}$, there exists a function $f_{j,\dir}$ which spans the space of microlocal solutions on $\gamma_{j,\dir}$ and admits the asymptotics \eqref{eq:mlBasis-f} 
with smooth functions $a_{j,\dir}$ and $b_{j,\dir}$ such that
\be\label{eq:mlBasis-fab2}
\begin{aligned}
&a_{1,\dir}=1+\ord(\til{\mu}_1^2),
&&b_{1,\dir}=\ord(\e_2),\\
&a_{2,\dir}=\ord(\e_1),
&&b_{2,\dir}=1+\ord(\til{\mu}_1^2).
\end{aligned}\qquad(\dir=\ell,r)
\ee
Such a function $f_{j,\dir}$ satisfying \eqref{eq:mlBasis-f} and \eqref{eq:mlBasis-fab2} is unique up to an error of $\ord(\til{\mu}_1^2)$. 
Again by Proposition~\ref{prop:2dim}, we can define the transfer matrix $T$ modulo $\ord(\til{\mu}_1^2)$ by \eqref{eq:Def-T}.
For $l\in\N$, define $\mu_{m,l}$ (a generalization of \eqref{eq:mu-k} with $\mu_m=\mu_{m,0}$) by
\be
\mu_{m,l}=\mu_{m,l}(\til{\e},h):=
\left\{
\begin{aligned}
&\til{\e}h^{-\frac{m-l}{m+1}}
&&\text{when }2l+1<m,\\
&
\til{\e}h^{-\frac12}(\log(1/h))^{\frac 12 \dl_{2l+1,m}}
&&\text{when }2l+1\ge m.
\end{aligned}\right.
\ee
We here used the Kronecker delta $\dl_{k,l}=1$ $(k=l)$ and $\dl_{k,l}=0$ $(k\neq l)$. 

\begin{theorem}\label{thm:Asymptot-T2}
Under Conditions~\ref{C1} and \ref{C2}, there exist $h_0>0$ and $\mu^0>0$ such that $T=T(h,\e_1,\e_2)$ admits 
\be\label{eq:Asym-T1}
T=\begin{pmatrix}
1+\ord(\mu_{m,\til{n}}^2)
&t_{12}
\\t_{21}
&1+\ord(\mu_{m,\til{n}}^2)
\end{pmatrix}
\ee
with
\be\label{eq:t12-t21}
\begin{aligned}
&t_{12}=-i
\e_1h^{-\frac{m-n_1}{m+1}}
(\til{\omega}_{m,n_1}(h;U_1,V_1-V_2)+\ord(\mu_{m,\til{n}}^2)),
\\&
t_{21}=-i
\e_2h^{-\frac{m-n_2}{m+1}}
(\til{\omega}_{m,n_2}(h;U_2,V_2-V_1)+\ord(\mu_{m,\til{n}}^2)),
\end{aligned}
\ee
uniformly for $0<h<h_0$ and $0<\e_j\le 1$ $(j=1,2)$ with $0<\mu_{m,\til{n}}<\mu^0$. 
Here, we define $\til{\omega}_{m,n}(h;W,Q)$ by
\be\label{eq:Osci-Int}
\til{\omega}_{m,n}(h;W,Q):=h^{\frac{n-1}{m+1}}\int_I\chi(x)W(x)\exp\left(\frac ih\int_0^xQ(y)dy\right)dx,
\ee
with $\chi\in C_c^\infty(I;[0,1])$ such that $\chi(x)=1$ near $x=0$. 
The values $\til{\omega}_{m,n_1}(h;U_1,V_1-V_2)$ and $\til{\omega}_{m,n_2}(h;U_2,V_2-V_1)$ are bounded as $h\to0^+$, and they are invariant up to $\ord(h^\infty)$ with respect to the choice of such a $\chi$. 
Their asymptotic behavior is given by \eqref{eq:Def-til-omega}, \eqref{eq:mn-even-tilomega0} and \eqref{eq:mn-odd-tilomega0}.
\end{theorem}

Note that the error $\ord(\mu_{m,\til{n}}^2)$ in the above formula is larger than or the same order as the essential error of order $\til{\mu}_1^2$ 
coming from the choice of the bases $(f_{1,\ell},f_{2,\ell},f_{1,r},f_{2,r})$. 

By applying the degenerate stationary phase method (see Lemma~\ref{lem:D-S-Ph}), the principal term admits the approximation
\be\label{eq:Def-til-omega}
\til{\omega}_{m,n_1}(h;U_1,V_1-V_2)=\til{\omega}_{m,n_1}^0(U_1,V_1-V_2)\left(1+\ord\left(h^{\frac{1}{m+1}}\right)\right),
\ee
with
\be\label{eq:mn-even-tilomega0}
\til{\omega}^0_{m,n}(W,Q)=\frac{2\eta_{m,n}W^{(n)}(0)}{(m+1)n!}\left(\frac{(m+1)!}{\left|Q^{(m)}(0)\right|}\right)^{\frac{n+1}{m+1}}\bm{\Gamma}\left(\frac{n+1}{m+1}\right),
\ee
when at least one of $m$ and $n_1$ is even. 
When both of $m$ and $n_1$ are odd, we have
\be
\til{\omega}_{m,n_1}(h;U_1,V_1-V_2)=h^{\frac1{m+1}}\til{\omega}_{m,n_1}^0(U_1,V_1-V_2)\left(1+\ord\left(h^{\frac2{m+1}}\right)\right),
\ee
with
\be\label{eq:mn-odd-tilomega0}
\begin{aligned}
\til{\omega}_{m,n}^0(W,Q)
&=\frac{2\eta_{m,n}}{(m+1)(n+1)!}\left(\frac{(m+1)!}{\left|Q^{(m)}(0)\right|}\right)^{\frac{n+2}{m+1}}\bm{\Gamma}\left(\frac{n+2}{m+1}\right)\\
&\quad\times\left(W^{(n+1)}(0)-\frac{(n+1)(n+2)Q^{(m+1)}(0)}{(m+1)(m+2)Q^{(n)}(0)}W^{(n)}(0)\right).
\end{aligned}
\ee
Here, $\bm{\Gamma}$ stands for the Gamma function, and 
\be\label{eq:eta}
\eta_{m,n}:=\left\{
\begin{aligned}
&i^n\cos\left(\frac{(1-mn)\pi}{2(m+1)}\right)&&\text{when $m$ is even},\\
&\exp\left(\frac{\ope{sgn}(Q^{(n)}(0))i(n+1)\pi}{2(m+1)}\right)&&\text{when $m$ is odd and $n$ is even},\\
&\exp\left(\frac{\ope{sgn}(Q^{(n)}(0))i(n+2)\pi}{2(m+1)}\right)&&\text{when $mn$ is odd}. 
\end{aligned}\right.
\ee
Note that $\eta_{m,n}$ vanishes if $m$ is even and $(n+1)m\in2(m+1)\N$. The other part of $\til{\omega}_{m,n}^0(W,Q)$ does not vanish if $Q^{(m)}(0)W^{(n)}(0)\neq0$ when at least one of $m$ and $n$ is even. 
The asymptotic behavior of $\til{\omega}_{m,n_2}(h;U_2,V_2-V_1)$ is given in the same manner.

\subsection{Unitarity of the local scattering matrix}
We define the local scattering matrix in an analogous way to that for the transversal crossings \cite{CdV3}. 
We can define it as an (asymptotically) unitary matrix in several cases such that the system is (equivalently) given by an operator with Hermitian matrix-valued symbol. 
Note that for each scalar symbol $p_j(x,\xi)$, the orientation of the Hamiltonian flow associated with $p_j$ is opposite to that with $-p_j$. However, the equations $p_j^w(x,hD_x)u=0$ and $-p^w(x,hD_x)u=0$ are equivalent. 
The Hermitian-valuedness determines the suitable orientation of the flow on each curve of the characteristic set. 

In the rest of this section, we always use the tuple $(f_{1,\ell},f_{2,\ell},f_{1,r},f_{2,r})$ such that $\det T=1$ (see \eqref{eq:DetT} for the existence of such a tuple) for simplicity. 
According to \eqref{eq:ess-error2}, the unitarity for the local scattering matrix corresponding to this tuple implies the asymptotic unitarity for any choice of the tuple, that is,
\ben
\left\|Sv\right\|_{\C^2}=(1+\ord(\til{\mu}_1^2))\left\|v\right\|_{\C^2}\quad v\in\C^2
\een
holds for the local scattering matrix $S$ corresponding to an arbitrary tuple.

Let us first consider the case that the full-symbol of $P$ itself,
\be
p(x,\xi)
=p_0(x,\xi)+\begin{pmatrix}0&\e_1 U_1(x)\\\e_2U_2(x)&0\end{pmatrix}
=\begin{pmatrix}\xi+V_1(x)&\e_1 U_1(x)\\\e_2U_2(x)&\xi+V_2(x)\end{pmatrix},
\ee
is Hermitian matrix-valued. It is equivalent to
\be\label{eq:Hermite1}
U_1=\overline{U_2}=:W,\quad\e_1=\e_2=:\e.
\ee
In this case, for $j=1,2$, the Hamiltonian flow associated with $p_j$ is incoming to $\rho_0$ on $\gamma_{j,\ell}$, and outgoing from $\rho_0$ on $\gamma_{j,r}$. Then the transfer matrix itself becomes the local scattering matrix $S=S(h,\e):=T(h,\e,\e)$.
\begin{proposition}\label{thm:Unitary-SM1}
Under Conditions~\ref{C1}, \ref{C2},  and \eqref{eq:Hermite1}, $S$ is unitary for $\mu_{m,n}$ small enough.
\end{proposition}

\begin{figure}
\centering
\includegraphics[bb=0 0 371 198, width=8cm]{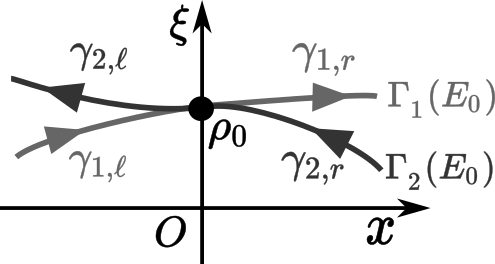}
\caption{Hamiltonian flows associated with the diagonal part of \eqref{eq:Modified-Sym}}
\label{Fig:}
\end{figure}

There is another case that System \eqref{eq:System} is equivalently given by a pseudodifferential operator with a Hermitian matrix-valued symbol. Consider the case with
\be\label{eq:Hermite2}
U_1=-\overline{U_2}=:W,\quad\e_1=\e_2=:\e.
\ee
Then the system \eqref{eq:System} is equivalent to
\be
\til{P}w:=\begin{pmatrix}1&0\\0&-1\end{pmatrix}Pw=\begin{pmatrix}hD_x+V_1&\e W\\\e\overline{W}&-(hD_x+V_2)\end{pmatrix}w=0,
\ee
where the full-symbol of $\til{P}$, 
\be\label{eq:Modified-Sym}
\til{p}(x,\xi)=\begin{pmatrix}\xi+V_1(x)&\e W(x)\\\e\overline{W(x)}&-(\xi+V_2(x))\end{pmatrix},
\ee
is Hermitian matrix-valued. 
In this case, the Hamiltonian flow associated with $p_1$ behaves in the same manner as the case with \eqref{eq:Hermite1}, 
whereas that with $-p_2$ is reverse to that with $p_2$, outgoing from $\rho_0$ on $\gamma_{2,\ell}$, and incoming to $\rho_0$ on $\gamma_{2,r}$. 
Taking this into account, we define $S$ by
\be\label{eq:Def-tilT}
\begin{pmatrix}
\alpha_{1,r}\\\alpha_{2,\ell}
\end{pmatrix}
=
S
\begin{pmatrix}
\alpha_{1,\ell}\\\alpha_{2,r}
\end{pmatrix},\quad
S
=\frac{1}{t_{22}}
\begin{pmatrix}
1
&t_{12}\\
-t_{21}&1
\end{pmatrix},
\ee
where $\alpha_{j,\dir}$ $(j,\dir)\in\{1,2\}\times\{\ell,r\}$ are determined for each microlocal solution $\psi$ near $\rho_0$ by \eqref{eq:mlCoeff}, and $t_{jk}$ stands for the $(j,k)$-entry of $T$.

\begin{proposition}\label{thm:Unitary-SM2}
Under Conditions~\ref{C1}, \ref{C2}, and \eqref{eq:Hermite2}, $S$ is unitary for $\mu_{m,n}$ small enough.
\end{proposition}

\begin{remark}
By substituting one of \eqref{eq:Hermite1} and \eqref{eq:Hermite2} into the asymptotic formula of $T$ given in Theorem~\ref{thm:Asymptot-T}, one has
\be
S=
\ope{Id}_2-i\e h^{-\frac{m-n}{m+1}}\left(T_{\ope{sub}}+\ord\left(h^{\frac{1}{m+1}}\right)\right)
+\ord\left(\mu_{m,n}^2
\begin{pmatrix}
1&\e h^{-\frac{m-n}{m+1}}\\\e h^{-\frac{m-n}{m+1}}&1
\end{pmatrix}\right)
\ee
with the constant $2\times2$-matrix 
\be
T_{\ope{sub}}=\begin{pmatrix}0&\overline{\omega_{m,n}}\\
\omega_{m,n}&0\end{pmatrix},
\quad
\omega_{m,n}=\til{\omega}_{m,n}^0(\overline{W},V_2-V_1).
\ee
\end{remark}

\begin{remark}\label{rem:Unitary-SM3}
Without the second condition $\e_1=\e_2$ in \eqref{eq:Hermite1} or in \eqref{eq:Hermite2}, we have an equivalence of System \eqref{eq:System} with that given by a pseudodifferential operator with Hermitian matrix-valued symbol
\ben
\begin{pmatrix}
\xi+V_1(x)&\e_1 W(x)\\
\e_1 \overline{W(x)}&\pm\frac{\e_1}{\e_2}(\xi+V_2(x))
\end{pmatrix}.
\een
In this case, we again have the unitarity of the local scattering matrix by replacing the tuple $(f_{1,\ell},f_{2,\ell},f_{1,r},f_{2,r})$ of bases by
\ben
\left(\sqrt[4]{\frac{\e_2}{\e_1}}\,f_{1,\ell},\,\sqrt[4]{\frac{\e_1}{\e_2}}\,f_{2,\ell},\,\sqrt[4]{\frac{\e_2}{\e_1}}\,f_{1,r},\,\sqrt[4]{\frac{\e_1}{\e_2}}\,f_{2,r}\right).
\een
\end{remark}

\section{Linear relationship between exact solutions to the system}\label{Sec:Proof1}
In this section, we solve System \eqref{eq:System} locally in an interval $I=]x_\ell,x_r[$ near $x=0$ $(x_\ell<0<x_r)$, and prove Proposition~\ref{thm:T} which gives the asymptotic behavior of the transfer matrix between two bases 
of exact solutions to the system. The argument in this section is not microlocal, and we will discuss on the correspondence between the exact solutions and the microlocal solutions in the next section. We obtain  Theorems~\ref{thm:Asymptot-T} and \ref{thm:Asymptot-T2} by combining Proposition~\ref{thm:T} and the microlocal argument in the next section.

We assume that $U_1,$ $U_2\in C_c^\infty(I)$ without loss of generalities, in addition to Conditions~\ref{C1} and \ref{C2}. In other words, we consider the operator $P$ with $U_1$ and $U_2$ replaced by $\chi U_1$ and $\chi U_2$, where $\chi\in C_c^\infty(I;[0,1])$ stands for a cut-off function such that $\chi(x)=1$ identically near $x=0$. Then solutions to the equation $Pw=0$ also solves the original equation near $x=0$ where $\chi=1$ identically.
We construct two bases $(w_{1,\ell},w_{2,\ell})$ and $(w_{1,r},w_{2,r})$ of solutions to the system in $I$ admitting the initial value 
\be\label{eq:Charactelization}
(w_{1,\dir},w_{2,\dir})\big|_{x=x_\dir}=
\exp\left(-\frac ih\int_0^{x_\dir}\ope{diag}(V_1(x),V_2(x))\,dx\right)\qquad
(\dir=\ell,r).
\ee
at $x=x_\ell$ and $x=x_r$, respectively, and compute the transfer matrix $T_{\ope{ex}}$  between these two bases (Proposition~\ref{thm:T}). 
We will show in Proposition~\ref{prop:w-f} that $w_{j,\dir}$ behaves as a basis of the space of microlocal solutions on $\gamma_{j,\dir}$ satisfying the characterization conditions \eqref{eq:mlBasis-f} and \eqref{eq:mlBasis-fab}, and it is microlocally zero near $\gamma_{\hat j,\dir}$ ($\hat j$ stands for the other index than $j$ out of $1$ and $2$). This with the fact that the space of microlocal solutions near $\rho_0$ is two-dimensional (Proposition~\ref{prop:2dim}) implies that $T_{\ope{ex}}$ is a representative of the equivalence class of $T$ in Theorems~\ref{thm:Asymptot-T} and \ref{thm:Asymptot-T2}.
Recall that the Wronskian $\cW(w_1,w_2)=\det(w_1,w_2)$ of any two solutions $w_1$ and $w_2$ to \eqref{eq:System} satisfies the ordinary differential equation $u'=(\ope{tr}V)u$, and it follows from \eqref{eq:Charactelization} that 
\be\label{eq:DetT}
\det T_{\ope{ex}}=1.
\ee

\begin{proposition}\label{thm:T}
Under the same condition as Theorem~\ref{thm:Asymptot-T}, the transfer matrix $T_{\ope{ex}}$ between the bases $(w_{1,\ell},w_{2,\ell})$ and $(w_{1,r},w_{2,r})$, i.e., $(w_{1,\ell},w_{2,\ell})=(w_{1,r},w_{2,r})T_{\ope{ex}}$, admits the same asymptotic behavior as $T$ given in Theorems~\ref{thm:Asymptot-T} and \ref{thm:Asymptot-T2}. Moreover, since we take a particular bases,  a better estimate
\ben
t_{11}-1=\ord(\min\{\mu_{m,\til{n}}^2,\til{\mu}_1^2h^{-\frac{m-n_2-\iota(mn_2)}{m+1}}\}),\quad
t_{22}-1=\ord(\min\{\mu_{m,\til{n}}^2,\til{\mu}_1^2h^{-\frac{m-n_1-\iota(mn_1)}{m+1}}\}),
\een
makes sense.  We here put $\iota(k)=0$ for $k$ even and $\iota(k)=1$ for $k$ odd. This estimate is strictly better when one of $n_2>m$ ($n_1>m$ for $t_{22}$) and $n_1=n_2=m$ holds. 
\end{proposition}

\subsection{Construction of the two bases of solutions}
We apply the idea of construction due to \cite{AFH1,FMW1} for a matrix Schr\"odinger operator to our problem. 
Our construction is much simpler than theirs. 

Let us fix a solution
\ben
u_j(x)=\exp\left(-\frac ih\int_0^xV_j(y)dy\right)\quad(j=1,2)
\een
to the scalar homogeneous equation $(hD_x+V_j)u_j=0$. It is well-known that for $j=1,2$,
the integral operator 
\be
K_jf(x):=\frac ih u_j(x)\int_{x_j}^x \frac{f(y)}{u_j(y)}dy\qquad(f\in C(I))
\ee
with $x_j\in\bar{I}$ gives a fundamental solution with the initial condition $K_jf(x_j)=0$.
%
We define the solutions $w_1$ and $w_2$ as follows
\be\label{eq:Def-w1}
w_1=w_1(x_1,x_2):=
\begin{pmatrix}
\displaystyle\sum_{k\ge0}(\e_1\e_2 K_1U_1K_2U_2)^k u_1\\
\dip-\e_2 K_2U_2\sum_{k\ge0}(\e_1\e_2 K_1U_1K_2U_2)^k u_1
\end{pmatrix},
\ee
and 
\be\label{eq:Def-w2}
w_2=w_2(x_1,x_2):=
\begin{pmatrix}
\dip-\e_1 K_1U_1\sum_{k\ge0}(\e_1\e_2 K_2U_2K_1U_1)^k u_2\\
\dip\sum_{k\ge0}(\e_1\e_2 K_2U_2K_1U_1)^k u_2
\end{pmatrix}.
\ee
Put
\be\label{eq:DefExSol}
w_{j,\ell}:=w_j(x_\ell,x_\ell),\quad
w_{j,r}:=w_j(x_r,x_r)\quad\text{for $j=1,2$.}
\ee
The following proposition shows the convergence of the infinite series in \eqref{eq:Def-w1} and \eqref{eq:Def-w2} for small $\mu_{m,\til{n}}>0$.

\begin{proposition}\label{prop:ConvS}
For any $x_1$, $x_2\in \bar{I}$, each one of the infinite sums 
\ben
\sum_{k\ge0}(\e_1\e_2K_1U_1K_2U_2)^k u_1\quad\text{and}\quad
\sum_{k\ge0}(\e_1\e_2 K_2U_2K_1U_1)^k u_2
\een
converges uniformly to a $C^1$-function when  $h>0$ and $\mu_{m,\til{n}}>0$ are small enough. 
\end{proposition}

To show the convergence for small $\mu_{m,\til{n}}$, we need a sharper estimate for oscillatory integrals than the estimates in \cite{AFH1,FMW1} where they considered the case with $\e_1=\e_2=h$ (remark also that their operator is of second order).
\begin{lemma}\label{lem:EstiOsci}
Let $I\subset\R$ be a compact interval containing $0$ in its interior, and let $a_h\in C^1(I)$, $\phi\in C^\infty(I;\R)$, $l_1,\,l_2\in\N=\{0,1,2,\ldots\}$. Suppose that $\phi'$ vanishes of $k$-th order $(k\ge1)$ at $x=0$ and does not vanish at any other point in $I$, and that $x^{-l_2}a_h'(x)$ is bounded on $I$ for each $h>0$. Then there exists $C>0$ independent of $h$, $\phi$, $a_h$ such that
\ben
\begin{aligned}
&\frac 1h\left|\int_I x^{l_1} a_h(x)e^{i\phi(x)/h}dx\right|\\
&\le
C\left(\sup_I\left|a_h\right|h^{-\left(\frac{k-l_1}{k+1}\right)_+}
+\sup_I\left|x^{-l_2}a_h'\right|
\left(\log\frac1h\right)^{\dl_{l_1+l_2+1,k}}
h^{-\left(\frac{k-l_1}{k+1}-\frac{l_2+1}{k+1}\right)_+}
\right)\\
&\le \left\{
\begin{aligned}
&Ch^{-\frac{k-l_1}{k+1}}
\left(\sup_I\left|a_h\right|+\left(\log\frac1h\right)^{\dl_{l_1+l_2+1,k}}h^{\frac{l_2+1}{k+1}}\sup_I\left|x^{-l_2}a_h'\right|\right)
&&\text{when }l_1+l_2+1\le k,
\\
&C
\left(h^{-\frac{k-l_1}{k+1}}\sup_I\left|a_h\right|+\sup_I\left|x^{-l_2}a_h'\right|\right)
&&\text{when }l_1+l_2\ge k \ge l_1,
\\
&C\left(\sup_I\left|a_h\right|+\sup_I\left|x^{-l_2}a_h'\right|\right)
&&\text{when }l_1>k,
\end{aligned}\right.
\end{aligned}
\een
uniformly for small $h>0$.
\end{lemma}

\begin{remark}\label{rem:repW}
Lemma~\ref{lem:EstiOsci} still holds even if $x^{l_1}$ is replaced by a $C^1$-function independent of $h$ which vanishes at $l_1$-th order at $x=0$. The functions $U_1$ or $U_2$ will take the place of $x^{l_1}$.
\end{remark}

\begin{proof}
Let $\chi_h\in C_c^\infty(I;[0,1])$ be such that $\chi_h(x)=1$ for $|x|\le h^{1/(k+1)}$ and $\chi_h(x)=0$ for $|x|\ge2h^{1/(k+1)}$. We divide the integral into two parts:
\be
\mc{I}_0:=\int_I \chi_h(x)x^{l_1}a_h(x)e^{i\phi(x)/h}dx,\qquad
\mc{I}_1:=\int_I (1-\chi_h(x))x^{l_1}a_h(x)e^{i\phi(x)/h}dx.
\ee
Then we have
\be\label{eq:EstiI0}
\left|\mc{I}_0\right|\le\sup_I\left|a_h\right|\int_{|x|\le2h^{1/(k+1)}}|x|^{l_1} dx
=C_{l_1}\sup_I\left|a_h\right| h^{\frac{l_1+1}{k+1}}.
\ee
On the support of $1-\chi_h$, the equality
\ben
e^{i\phi(x)/h}=\left(\frac h{i\phi'(x)}\right) \frac {d}{dx}e^{i\phi(x)/h}
\een
is valid since $\phi'$ does not vanish. 
An integration by parts shows the estimate
\ben
\left|\mc{I}_1\right|\le h\int_I\left|\frac{d}{dx}\left((1-\chi_h(x))\frac{x^{l_1}}{\phi'(x)}a_h(x)\right)\right|dx.
\een
One can impose also $\left|\chi_h'\right|\le h^{-1/(k+1)}$ on $\chi_h$. Taylor's formula shows $\phi'(x)=(k!)^{-1}\phi^{(k+1)}(0)x^{k}(1+\ord(x))$ for small $x$. By assumption, $\left|\phi'(x)\right|\ge c$ for some $c>0$ outside an $h$-independent neighborhood of $x=0$. Hence, we obtain
\be\label{eq:EstiI11}
\begin{aligned}
\left|\mc{I}_1\right|\le 
&Ch
\int_{|x|\ge h^{1/(k+1)}}\left(|x|^{l_1-k-1}\sup\left|a_h\right|+|x|^{l_1+l_2-k}\sup_I\left|x^{-l_2}a_h'\right|\right)dx
\\&\qquad
+C\sup\left|a_h\right|h^{\frac k{k+1}}\int_{\ope{supp}\chi_h'}|x|^{l_1-k}dx
\end{aligned}
\ee
when $l_1\neq k$, and 
\be\label{eq:EstiI12}
\left|\mc{I}_1\right|\le Ch
\int_{|x|\ge h^{1/(k+1)}}\left(\sup\left|a_h\right|+|x|^{l_2}\sup_I\left|x^{-l_2}a_h'\right|\right)dx
+C\sup\left|a_h\right|h^{\frac k{k+1}}\int_{\ope{supp}\chi_h'}dx.
\ee
when $l_1=k$. The required estimate follows from \eqref{eq:EstiI0}, \eqref{eq:EstiI11} and \eqref{eq:EstiI12}.
\end{proof}

\begin{proof}[Proof for Proposition~\ref{prop:ConvS}]
Let $C_n^1(I)$ $(n\in\N)$ be the set of continuously differentiable functions $f$ on $I$ such that $x^{-n}f'$ is bounded. We define $\|\cdot\|_{n,q}$ for $q\in\R$ by
\be
\|f\|_{n,q}:=\sup_I|f|+h^q\sup_I|x^{-n}f'|\qquad(f\in C^1_n(I)).
\ee
Let $f\in C_{n_1}^1(I)$ and $g\in C_{n_2}^1(I)$. 
By applying Lemma~\ref{lem:EstiOsci}, we obtain the following estimates
\be\label{eq:OpEstis}
\begin{aligned}
&\left\|(u_1)^{-1}K_1U_1K_2U_2(u_1f)\right\|_{n_1,\frac{n_1+1}{m+1}}
\le C\left(\log\frac1h\right)^{\dl_{n_1+n_2+1,m}}h^{-\frac{2m-n_1-n_2}{m+1}}
\|f\|_{n_1,\frac{n_1+1}{m+1}}
,\\
&\left\|(u_2)^{-1}K_2U_2K_1U_1(u_2g)\right\|_{n_2,\frac{n_2+1}{m+1}}
\le C\left(\log\frac1h\right)^{\dl_{n_1+n_2+1,m}}h^{-\frac{2m-n_1-n_2}{m+1}}
\|g\|_{n_2,\frac{n_2+1}{m+1}},
\end{aligned}
\ee
when $n_1+n_2\le m-1$, and 
\be
\begin{aligned}
&\left\|(u_1)^{-1}K_1U_1K_2U_2(u_1f)\right\|_{n_1,q_1}
\le Ch^{-1}
\|f\|_{n_1,q_1},
\\
&\left\|(u_2)^{-1}K_2U_2K_1U_1(u_2g)\right\|_{n_2,q_2}
\le Ch^{-1}
\|g\|_{n_2,q_2},
\end{aligned}
\ee
for any $q_1$ and $q_2$ satisfying
\be\label{eq:q1q2Cond}
\left\{
\begin{aligned}
&\frac{m-n_2}{m+1}\le q_1\le\frac{n_1+1}{m+1},
&&\frac{m-n_1}{m+1}\le q_2\le\frac{n_2+1}{m+1}
&&\text{when }n_1+n_2\ge m\ge n_1,\, n_2,
\\
&\frac{m-n_2}{m+1}\le q_1\le1,
&&0\le q_2\le \frac{n_2+1}{m+1}
&&\text{when }n_1> m\ge n_2,
\\
&0\le q_1\le1,
&&0\le q_2\le 1
&&\text{when }n_1,\,n_2>m.
\end{aligned}\right.
\ee
We here only show \eqref{eq:OpEstis}. The differences in $q_1$ and $q_2$ come only from those in Lemma~\ref{lem:EstiOsci}.
We have by definition
\ben
(u_2)^{-1}K_2U_2(u_1f)(x)=\frac ih\int_{x_2}^x \exp\left(\frac ih\int_0^y(V_2(t)-V_1(t))dt\right)U_2(y)f(y)dy.
\een
Thus, Lemma~\ref{lem:EstiOsci} (and Remark~\ref{rem:repW}) with $(k,l_1,l_2)=(m,n_2,n_1)$, $a_h=f$, $\phi(y)=\int_0^y(V_2(t)-V_1(t))dt$ is  applicable;
\be\label{eq:KU-esti1}
\sup_I \left|(u_2)^{-1}K_2U_2(u_1f)\right|\le C
h^{-\frac{m-n_2}{m+1}}\left(\sup_I\left|f\right|+\left(\log\frac1h\right)^{\dl_{n_1+n_2+1,m}}h^{\frac{n_1+1}{m+1}}\sup_I\left|x^{-n_1}f'\right|\right).
\ee
For the derivative, we have
\ben
\frac d{dx}\left((u_2)^{-1}K_2U_2(u_1f)\right)(x)=\frac ih \exp\left(\frac ih\int_0^x(V_2(t)-V_1(t))dt\right)U_2(x)f(x),
\een
and consequently
\be\label{eq:KU-esti2}
h^{\frac{n_2+1}{m+1}}\sup_I\left|x^{-n_2}\frac d{dx}\left((u_2)^{-1}K_2U_2(u_1f)\right)\right|\le Ch^{-\frac{m-n_2}{m+1}}\sup_I|x^{-n_2}U_2|\sup_I |f|.
\ee
The inequalities \eqref{eq:KU-esti1} and \eqref{eq:KU-esti2} with the boundedness of $\sup_I|x^{-n_2}U_2|$ gives the estimate 
\ben
\left\|(u_2)^{-1}K_2U_2(u_1f)\right\|_{n_2,\frac{n_2+1}{m+1}}\le C\left(\log\frac1h\right)^{\dl_{n_1+n_2+1,m}}h^{-\frac{m-n_2}{m+1}}\left\|f\right\|_{n_1,\frac{n_2+1}{m+1}}.
\een
The estimate \eqref{eq:OpEstis} is proved by repeating a symmetric process. Note that the logarithmic factor appears only in \eqref{eq:KU-esti1}, which is multiplied to the supremum of the derivative. Note also that this factor does not appear in the estimate \eqref{eq:KU-esti2} of the derivative. Therefore the logarithmic factor is multiplied at most once in \eqref{eq:OpEstis} even the above process is repeated twice.
Since $(u_2)^{-1}K_2U_2u_1$ belongs to $C_{n_2}^1(I)$ and $(u_1)^{-1}K_1U_1u_2$ to $C_{n_1}^1(I)$, the infinite series converge when $h$ and $\mu_{m,\til{n}}$ are small enough.
\end{proof}

\subsection{Asymptotic behavior of the transfer matrix: proof of Proposition~\ref{thm:T}}
We compute asymptotic behavior of the transfer matrix $T_{\ope{ex}}$. 
For convenience, we introduce the following notations;
\begin{align*}
&\zeta(h):=\left(\frac{\mu_{m,\til{n}}}{\til{\e}}\right)^2
=\left(\log\frac1h\right)^{\dl_{2\til{n}+1,m}}h^{-\left(\frac{m-2\til{n}-1}{m+1}\right)_+-1},\quad
(k)_+=\max\{k,0\},\\\
&\zeta_j(h):=h^{-1-\nu(m,n_j)},\quad\nu(m,n_j)=\frac{m-n_j-\iota(mn_j)}{m+1},\quad
\iota(k)=\frac{1-(-1)^k}2.
\end{align*}
The following lemma gives a representation of $T_{\ope{ex}}$.
\begin{lemma}\label{lem:AandA0}
We have
\be\label{eq:rep-T}
A:=T_{\ope{ex}}-\ope{Id}_2=\frac1{ih}\ope{diag}(u_1(x_r),u_2(x_r))^{-1}\int_IF(x)B(x)G(x)dx
\ee
with the matrix valued functions $F$, $G$ and $B$ given by 
\begin{align*}
&F=\begin{pmatrix}\frac{\e_1}{u_1}U_1&0\\0&\frac{\e_2}{u_2}U_2\end{pmatrix},\qquad
B=\begin{pmatrix}-\e_2K_{2,\ell}U_2&1\\1&-\e_1K_{1,\ell}U_1\end{pmatrix},\\
&G=
\sum_{k\ge0}(\e_1\e_2)^k\ope{diag}\left((K_{1,\ell}U_1K_{2,\ell}U_2)^ku_1,
\,(K_{2,\ell}U_2K_{1,\ell}U_1)^ku_2\right).
\end{align*}
Moreover, $A$ admits the asymptotic formula
\be\label{eq:asymptotics-A}
A=A_0+\ord
\left(\e_1\e_2
\begin{pmatrix}\min\{\zeta(h),\zeta_2(h)\}&\e_1 h^{-\nu(m,n_1)}\zeta(h)\\\e_2 h^{-\nu(m,n_1)}\zeta(h)&\min\{\zeta(h),\zeta_1(h)\}\end{pmatrix}\right),
\ee
where 
$A_0$ is given by
\ben
A_0=\frac1{ih}\int_IU(\e_1,\e_2,x)\ope{diag}\left(\frac{u_1(x)}{u_2(x)},\frac{u_2(x)}{u_1(x)}\right)dx.
\een
\end{lemma}

Recalling the definition \eqref{eq:Def-til-omega} of $\til{\omega}_{m,n}(W,Q)$, we obtain the following representation of $A_0$:
\ben
A_0=-i\begin{pmatrix}0&\e_1 h^{-\frac{m-n_1}{m+1}}\til{\omega}_{m,n_1}(U_1,V_1-V_2)\\
\e_2 h^{-\frac{m-n_2}{m+1}}\til{\omega}_{m,n_2}(U_2,V_2-V_1)&0
\end{pmatrix}.
\een
This with Lemma~\ref{lem:AandA0} ends the proof of Proposition~\ref{thm:T}.

\begin{proof}
Recall that by definition of $K_{j,\ell}$, we have 
\ben
K_{j,\ell}f(x_r)=\frac ih u_j(x_r)\int_{I}\frac{f(y)}{u_j(y)}dy\qquad(j=1,2).
\een
Substituting this into the definition of $w_{j,\ell}$ $(=1,2)$, we immediately obtain the equality
\ben
(w_{1,\ell}(x_r),w_{2,\ell}(x_r))=
\ope{diag}(u_1(x_r),u_2(x_r))(\ope{Id}_2+A).
\een
This with \eqref{eq:Charactelization} implies the identity \eqref{eq:rep-T} for $A$. 

Let us estimate the remainder terms. Put
\begin{align*}
&f_1:=(u_2)^{-1}K_2U_2u_1,\quad
f_2:=(u_1)^{-1}K_1U_1K_2U_2u_1=(u_1)^{-1}K_1U_1u_2f_1,\\
&f_3:=(u_2)^{-1}K_2U_2K_1U_1K_2U_2u_1=(u_2)^{-1}K_2U_2u_1f_2.
\end{align*}
We show the estimates 
\be\label{eq:RemainderEstis}
\left\|f_2\right\|_{n_1,q_1}=\ord(\min\{\zeta(h),\zeta_2(h)\}),\quad
\left\|f_3\right\|_{n_2,q_2}=\ord(h^{-\nu(m,n_2)}\zeta(h)),
\ee
for any $q_1$ and $q_2$ satisfying
\ben
q_1\ge\left\{
\begin{aligned}
&\frac{n_1+1}{m+1}&&\text{when }\ n_1+n_2+1\le m,\\
&\frac{m-n_2}{m+1}&&\text{when }\ n_1+n_2+1>m\ge n_2,\\
&0&&\text{when }\ n_2>m\ge n_2,
\end{aligned}\right.\qquad
q_2\ge\left\{
\begin{aligned}
&-\frac{m-n_2}{m+1}&&\text{when }\ n_2<m,\\
&0&&\text{when }\ n_2\ge m.
\end{aligned}
\right.
\een
Then the every term of the infinite sum is estimated inductively by the estimates \eqref{eq:OpEstis} of the operators between $C_{n_1}^1(I)$ and $C_{n_2}^1(I)$. Those estimates are applicable since the condition imposed on $q_1$ and $q_2$ here and that in \eqref{eq:q1q2Cond} have an intersection.

The degenerate stationary phase method (Lemma~\ref{lem:D-S-Ph}) gives us the estimate $\sup_I\left|f_1\right|=\ord(h^{-\nu (m,n_2)})$. 
Then for $f_2$, we have on one hand
\ben
\left|f_2\right|\le 
h^{-1}\sup_I\left|U_1f_1\right|=\ord(h^{-\nu (m,n_2)-1})=\ord(\zeta_2(h)).
\een
On the other hand, applying Lemma~\ref{lem:EstiOsci} with $\sup_I\left|x^{-n_2}f'_1(x)\right|=\ord(h^{-1})$, we also have
\ben
\left|f_2\right|=\ord\left(\left(\log\frac1h\right)^{\dl_{n_1+n_2+1,m}}h^{-\left(\frac{m-n_1-n_2-1}{m+1}\right)_+-1}\right)=\ord(\zeta(h)).
\een
Thus $\sup_I\left|f_2\right|=\ord(\min\{\zeta(h),\zeta_2(h)\})$ follows. The former one is better than (or the same as) the latter one when $n_2\le m$ whereas the latter is so when $n_2< m$. 
For the derivative, we have
\ben
\left|x^{-n_1}f'_2\right|\le 
h^{-1}\sup_I\left|f_1\right|=\ord(h^{
-\nu (m,n_2)-1})
=\ord(h^{-q_1}\min\{\zeta(h),\zeta_2(h)\}).
\een
This shows the estimate \eqref{eq:RemainderEstis} for $\|f_2\|_{n_1,q_1}$. 
For $\sup_I\left|f_3\right|$, we again compare the estimate $\sup_I\left|f_3\right|\le h^{-1}\sup_I\left|f_2\right|=\ord(h^{-1}\min\{\zeta(h),\zeta_2(h)\})$ with that obtained by applying Lemma~\ref{lem:EstiOsci}. Then the latter, $\sup_I\left|f_3\right|=\ord(h^{-\nu(m,n_2)}\zeta(h))$, is always better. We also have
\ben
\left|x^{-n_2}f_3'\right|\le h^{-1}\sup_I\left|f_2\right|
=\ord(h^{-1}\min\{\zeta(h),\zeta_2(h)\}).
\een
Thus, we finally obtain the estimate \eqref{eq:RemainderEstis} for $\|f_3\|_{n_2,q_2}$. 
\end{proof}

Let us recall the degenerate stationary phase (see e.g., \cite[Formulae (7.7.30) and (7.7.31)]{Ho}). 
\begin{lemma}\label{lem:D-S-Ph}
Let $a\in C_c^\infty(I)$, $\phi\in C^\infty(I)$ be independent of $h$, let $x=0$ be the unique zero in $I$ of $\phi'$, and let $k$ be its order. 
Then for any $N\in\N$, we have
\be\label{eq:DSP-o}
\int a(x)e^{i\phi(x)/h}dx=\frac2k\sum_{l<N}\frac{i^la_\phi^{(l)}(0)}{l!}\bm{\Gamma}\left(\frac{l+1}{k}\right)\cos\left(\frac{(1-(k-1)l)\pi}{2k}\right)h^{\frac{l+1}{k}}+\ord(h^{\frac{N+1}k})
\ee
as $h\to0^+$ for odd $k$, and 
\be\label{eq:DSP-e}
\int a(x)e^{i\phi(x)/h}dx=\frac2k\sum_{l<N}\frac{a_\phi^{(2l)}(0)}{(2l)!}\bm{\Gamma}\left(\frac{2l+1}{k}\right)e^{\ope{sgn}(\phi^{(k)}(0))i\pi (2l+1)/2k}h^{\frac{2l+1}{k}}+\ord(h^{\frac{2N+1}k})
\ee
as $h\to0^+$ for even $k$. Here, we define the function $a_\phi$ by
\ben
a_\phi(y(x))=\frac{a(x)}{y'(x)},\quad
y(x)=x\til{\phi}(x)^{1/k},\quad
\til{\phi}(x)=\frac{\ope{sgn}(\phi^{(k)}(0))}{x^k (k-1)!}\int_0^x (x-t)^{k-1}\phi^{(k)}(t)dt.
\een
\end{lemma}

\section{Completion of the proof}\label{Sec:Proof2}
In Subsection~\ref{sec:ml-justify}, we show that the matrix $T_{\ope{ex}}$, of which we have computed asymptotic behavior in the previous section, is a representative of the equivalence class of $T$ in Theorem~\ref{thm:Asymptot-T}, and that ends the proof of Theorem~\ref{thm:Asymptot-T}. 
In Subsection~\ref{sec:unitarity}, we prove the unitarity stated in Propositions~\ref{thm:Unitary-SM1} and \ref{thm:Unitary-SM2} of the local scattering matrix by using symmetries of the construction of the solutions.
\subsection{Microlocal connection formula}\label{sec:ml-justify}
In this subsection, we prove the following two propositions. They with Proposition~\ref{thm:T} imply Theorem~\ref{thm:Asymptot-T}. 
\begin{proposition}\label{prop:w-f}
The solutions $w_{1,\ell}$, $w_{2,\ell}$, $w_{1,r}$, and $w_{2,r}$ satisfy the conditions \eqref{eq:mlBasis-f} and \eqref{eq:mlBasis-fab} on the asymptotic behavior of the microlocal solution $f_{1,\ell}$, $f_{2,\ell}$, $f_{1,r}$, and $f_{2,r}$, respectively. They also satisfy
\be\label{eq:ml0}
w_{1,\dir}\equiv0\qtext{microlocally near}\gamma_{2,\dir}\quad\text{and }\ 
w_{2,\dir}\equiv0\qtext{microlocally near}\gamma_{1,\dir}\quad(\dir=\ell,r).
\ee
\end{proposition}
\begin{proposition}\label{prop:2dim}
Under Conditions~\ref{C1} and \ref{C2}, the space of microlocal solutions near $\rho_0$ is two-dimensional for $\mu_{m,\til{n}}$ small enough.
\end{proposition} 
\begin{proof}[Proof of Proposition~\ref{prop:w-f}]
We prove the proposition only for $w_{1,\ell}$. Proofs for the others are parallel. 
By construction, $w_{1,\ell}$ has the following form:
\ben
w_{1,\ell}=\begin{pmatrix}u_1+\e_1\e_2h^{-1}\beta(x;\e_1,\e_2,h)u_1(x)\\\e_2h^{-1}\alpha(x;\e_1,\e_2,h)u_2(x)\end{pmatrix},
\een
where $\alpha$ and $\beta$ are absolutely convergent infinite sums of the functions of the form
\ben
\frac1h\int_{x_\ell}^xf(y)\exp\left(\pm\frac ih\int_0^y(V_1(t)-V_2(t))dt\right)dy
\een
with $f\in C^1(I)$ such that $f$ and $f'$ are  bounded as $h\to0^+$. 
By an integration by parts, for any 
$\dl>0$, we have
\ben
\sup_{x_\ell\le x\le-\dl}\frac1h\left|\int_{x_\ell}^xf(y)\exp\left(\pm\frac ih\int_0^y(V_1(t)-V_2(t))dt\right)dy\right|
\le C\left(\sup_I\left(\left|f\right|+\left|f'\right|\right)\right)\le C,
\een
since $V_1-V_2$ does not vanish on $[x_\ell,-\dl]$. This implies that $w_{1,\ell}$ satisfies \eqref{eq:mlBasis-f} and \eqref{eq:mlBasis-fab}.
 
We then show that $w_{1,\ell}$ is microlocally 0 near $\gamma_{2,\ell}$. Again by construction, its behavior on $\{x<\min(\ope{supp}U_1\cup\ope{supp}U_2)\}$ is explicitly given by
\be
(w_{1,\ell}(x),w_{2,\ell}(x))=\ope{diag}(u_1(x),u_2(x)). 
\ee
This shows that in $\{(x,\xi)\in I\times\R;\,x<\min(\ope{supp}U_1\cup\ope{supp}U_2)\}$, the semiclassical wavefront set $\ope{WF}_h(w_{1,\ell})$ is a subset of $\gamma_{l,1}$. 
Then the invariance of the semiclassical wavefront set along Hamiltonian curve implies \eqref{eq:ml0}.
\end{proof}

\begin{figure}
\centering
\includegraphics[bb=0 0 349 227, width=8cm]{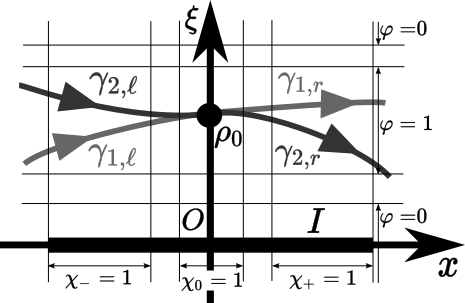}
\caption{Cut-off functions in the Proof of Proposition~\ref{prop:2dim}}
\label{Fig:CO}
\end{figure}
\begin{proof}[Proof of Proposition~\ref{prop:2dim}]
Let $w$ be a microlocal solution to the equation $P w= 0$ near  the crossing point $\rho_0=(0,-V_0)$. We construct an exact solution $\til{w}$ such that 
\ben
\til{w}\equiv w\quad\text{microlocally near }\rho_0.
\een
Then Proposition~\ref{prop:2dim} follows.
Before that, we first construct a quasi-mode $\psi$ such that
\ben
\|P\psi\|_{L^2(I')}=\ord(h^\infty),\quad
\psi\equiv w\quad\text{microlocally near }\rho_0,
\een
where $I'$ is a neighborhood of $x=0$ on which $\chi=1$ identically.
Take functions $\chi_-,\chi_0,\chi_+,\pphi\in C^\infty(\R;[0,1])$ such that $\pphi(\xi)=1$ near $\xi=-V_0$,
\be
\chi_-+\chi_0+\chi_+=\chi\ \text{ on }\R,\quad
\chi_0=1\ \text{ near }x=0,\quad
\chi_\pm=0\ \text{ for }\mp x>0.
\ee
We moreover can suppose that the support of $\chi_0$ and that of $\pphi$ are compact, and satisfy
\ben
(\ope{supp}\chi_0\times\ope{supp}\pphi)\cap\ope{WF}_h(w)\subset \Gamma\cap\{x\in I\}\subset I\times\{\pphi(\xi)=1\}.
\een
We also take solutions $w_l$ and $w_r$ such that
\ben
w\equiv w_\dir\qtext{microlocally near}\gamma_{1,\dir}\cup\gamma_{2,\dir}\quad(\dir=\ell,r).
\een
Such a solution $w_\dir$ exists since the space of microlocal solutions on each $\gamma_{j,\dir}$ is one dimensional. Put
\be
\psi(x):=\chi_0(x)\pphi(hD_x)\chi(x)w(x)+\chi_-(x) w_l+\chi_+(x)w_r.
\ee
This satisfies $\psi\equiv w$ microlocally near $\rho_0$. 
Let us show that this is a quasi-mode.
We have
\ben
P\psi=\chi_0[P,\pphi(hD_x)]\chi w+\chi_0\pphi(hD_x)P\chi w+[P,\chi_0]\pphi(hD_x)\chi w+[P,\chi_-]w_l+[P,\chi_+]w_r.
\een
We see that the first two terms are $\ord(h^\infty)$ since the essential support of the Weyl symbol of the pseudodifferential operator $\chi_0[P,\pphi(hD_x)]$ and of $\chi_0\pphi(hD_x)$ are compact, and does not intersect with the semiclassical wavefront set of $\chi w$ and of $P\chi w$, respectively. By the equality $[P,\chi_0]=-([P,\chi_-]+[P,\chi_+])$, the other terms are sum of the following two terms:
\begin{align*}
&[P,\chi_-](w_l-\pphi(hD_x)\chi w)=[P,\chi_-](1-\pphi(hD_x))w_l+\ord(h^\infty),\\
&[P,\chi_+](w_r-\pphi(hD_x)\chi w)=[P,\chi_+](1-\pphi(hD_x))w_r+\ord(h^\infty).
\end{align*}
Since $w_l$ and $w_r$ are exact solutions to $P w_\dir=0$ and microlocally zero away from $\Gamma$, the above two terms are also $\ord(h^\infty)$.

By applying the Cramer's rule, we have
\ben
\psi(x)=\frac1{\cW(w_{1,\ell},w_{2,\ell})(x)}\left(\cW(\psi,w_{2,\ell})(x)w_{1,\ell}(x)+\cW(w_{1,\ell},\psi)(x)w_{2,\ell}(x)\right).
\een
Since $\psi$ is a quasi-mode, we have
\ben
\frac{\cW(\psi,w_{2,\ell})(x)}{\cW(w_{1,\ell},w_{2,\ell})(x)}=\frac{\cW(\psi,w_{2,\ell})(0)}{\cW(w_{1,\ell},w_{2,\ell})(0)}+\ord(h^\infty),\quad
\frac{\cW(w_{1,\ell},\psi)(x)}{\cW(w_{1,\ell},w_{2,\ell})(x)}=\frac{\cW(w_{1,\ell},\psi)(0)}{\cW(w_{1,\ell},w_{2,\ell})(0)}+\ord(h^\infty),
\een
and consequently
\be
\til{w}=\psi+\ord(h^\infty)\qtext{with}
\til{w}:=\frac{\cW(\psi,w_{2,\ell})(0)}{\cW(w_{1,\ell},w_{2,\ell})(0)}w_{1,\ell}+
\frac{\cW(w_{1,\ell},\psi)(0)}{\cW(w_{1,\ell},w_{2,\ell})(0)}w_{2,\ell}.
\ee
This $\til{w}$ is what we aimed to construct.
\end{proof}
\subsection{Unitarity of the local scattering matrix}\label{sec:unitarity}
We here prove Propositions~\ref{thm:Unitary-SM1} and \ref{thm:Unitary-SM2}.
They are obtained by substituting 
 \eqref{eq:Hermite1} and \eqref{eq:Hermite2} into the definition of the matrix $A$.

We have $u_1=\overline{u_-}/u_+$ and $u_2=u_-/u_+$ with
\ben
u_{\pm}(x):=\exp\left(\frac{i}{2h}\int_0^x(V_1(t)\pm V_2(t))dt\right).
\een
We also have the symmetry
\be
u_+(x)K_{1,\dir}U_1f(x)=\overline{u_-(x)}\int_{x_\dir}^x \frac{W(y)}{u_1(y)}f(y)dy
=\overline{u_+(x)K_{2,\dir}U_2\overline{f}(x)},
\ee
under \eqref{eq:Hermite1}, and similarly
\be
u_+(x)K_{1,\dir}U_1f(x)
=-\overline{u_+(x)K_{2,\dir}U_2\overline{f}(x)},
\ee
under \eqref{eq:Hermite2}. 
By induction, we obtain
\ben
u_+w_{2,\dir}=\left\{
\begin{aligned}
&\begin{pmatrix}0&-1\\1&0\end{pmatrix}\overline{u_+w_{1,\dir}}
&&\text{under \eqref{eq:Hermite1},}\\
&\begin{pmatrix}0&1\\1&0\end{pmatrix}\overline{u_+w_{1,\dir}}
&&\text{under \eqref{eq:Hermite2},}
\end{aligned}\right.
\een
for $\dir=\ell,r$. This implies that $T$, more precisely one of its representative $T_{\ope{ex}}$, has the form
\be
T=\begin{pmatrix}\tau_1&-\overline{\tau_2}\\\tau_2&\overline{\tau_1}\end{pmatrix},\qquad
T=\begin{pmatrix}\tau_1&\overline{\tau_2}\\\tau_2&\overline{\tau_1}\end{pmatrix},
\ee
respectively under \eqref{eq:Hermite1}, \eqref{eq:Hermite2}. This with \eqref{eq:DetT}: $\det T=1$ directly implies that $T$ is unitary \eqref{eq:Hermite1}. By substituting this into \eqref{eq:Def-tilT}, we also obtain the unitarity of
\be
\til{T}=\frac1{\overline{\tau_1}}
\begin{pmatrix}
1&\overline{\tau_2}\\
-\tau_2&1
\end{pmatrix}.
\ee

\section*{Acknowledgements}
This work is supported by Grant-in-Aid for JSPS Fellows Grant Number JP22KJ2364. The author appreciates professors S.~Fujii\'e and T.~Watanabe for helpful discussions. A part of discussion was done during the conference held in the Research Institute for Mathematical Sciences, an International Joint Usage/Research Center located in Kyoto University.

\end{document}